%% file: main.tex
\documentclass[11pt]{llncs}

\usepackage[english]{babel}
\usepackage[T1]{fontenc}
\usepackage[latin1]{inputenc}
\usepackage{times}
\usepackage{amssymb,amsmath,amscd,latexsym,graphicx}
\usepackage[usenames]{color}

\usepackage[ruled]{algorithm2e}

\SetAlFnt{\small}
\SetAlCapFnt{\small}
\SetAlCapNameFnt{\small}
\SetAlCapHSkip{0pt}
\IncMargin{-\parindent}

\input{macros.tex}

\newtheorem{fact}{Fact}

\title{Profinite Techniques for Probabilistic Automata\\
and the Markov Monoid Algorithm\thanks{A preliminary version appeared in the proceedings of STACS'2016~\cite{Fijalkow16}.
This work was supported by The Alan Turing Institute under the EPSRC grant EP/N510129/1.}}

\author{Nathana{\"e}l Fijalkow}

\institute{University of Oxford, United Kingdom}

\begin{document}

\maketitle

\begin{abstract}
We consider the value $1$ problem for probabilistic automata over finite words:
it asks whether a given probabilistic automaton accepts words with probability arbitrarily close to $1$.
This problem is known to be undecidable.
However, different algorithms have been proposed to partially solve it;
it has been recently shown that the Markov Monoid algorithm, based on algebra, 
is the most correct algorithm so far.
The first contribution of this paper is to give a characterisation of the Markov Monoid algorithm.

The second contribution is to develop a profinite theory for probabilistic automata, called the prostochastic theory.
This new framework gives a topological account of the value $1$ problem,
which in this context is cast as an emptiness problem.
The above characterisation is reformulated using the prostochastic theory, allowing us to give a simple and modular proof.
\end{abstract}


\section{Introduction}
\input{intro}

\section*{Acknowledgments}
This paper and its author owe a lot to Szymon Toru{\'n}czyk's PhD thesis
and its author, to Sam van Gool for his expertise on Profinite Theory,
to Miko{\l}aj Boja{\'n}czyk for his insightful remarks and 
to Jean-{\'E}ric Pin for his numerous questions and comments.
The opportunity to present partial results on this topic in several scientific meetings
has been a fruitful experience, and I thank everyone that took part in it.
Last but not least, the reviewers greatly participated in improving the quality of this paper.

\section{Probabilistic Automata and the Value 1 Problem}
\label{sec:defs}
\input{defs}

\section{Characterisation of the Markov Monoid Algorithm}
\label{sec:mma}
\input{mma}

\section{The Prostochastic Theory}
\label{sec:prostochastic}
In this section, we introduce the prostochastic theory,
which draws from profinite theory to give a topological account of probabilistic automata.
We construct the free prostochastic monoid in Subsection~\ref{subsec:prostochastic:theory}.

The aim of this theory is to give a topological account of the value $1$ problem;
we show in Subsection~\ref{subsec:prostochastic:reformulation}
that the value $1$ problem can be reformulated as an emptiness problem for prostochastic words.

In Subsection~\ref{subsec:prostochastic:limit} we define the notion of polynomial prostochastic words.
%show how to construct non-trivial prostochastic words.

The Subsection~\ref{subsec:prostochastic:proof} is devoted to a technical proof, about the powers of stochastic matrices.

The characterisation given in Section~\ref{sec:mma} is stated and proved in this new framework in Subsection~\ref{subsec:prostochastic:reformulation_characterisation}.

\subsection{The Free Prostochastic Monoid}
\label{subsec:prostochastic:theory}
\input{theory}

\subsection{Reformulation of the Value 1 Problem}
\label{subsec:prostochastic:reformulation}
\input{reformulation}

\subsection{The Limit Operator, Fast and Polynomial Prostochastic Words}
\label{subsec:prostochastic:limit}
\input{limit}

\subsection{Powers of a Stochastic Matrix}
\label{subsec:prostochastic:proof}
\input{proof}

\subsection{Reformulating the Characterisation}
\label{subsec:prostochastic:reformulation_characterisation}
\input{reformulation_characterisation}

\section{Towards an Optimality Argument}
\label{sec:conclusions}

In this section, we build on the characterisation obtained above to argue that the Markov Monoid algorithm is in some sense optimal.

\subsection{Undecidability of the Two-Tier Value 1 Problem}
\label{subsec:conclusions:two-tier}
\input{undecidability}

\subsection{Combining the Two Results}
\label{subsec:conclusions:optimality}
\input{optimality}

\bibliographystyle{plain}
\bibliography{bib,perso}

\end{document}

%% file: macros.tex
\newcommand{\set}[1]{\{ #1 \}}
\newcommand{\val}[1]{\text{val}(#1)}

\def\qed{\rule{0.4em}{1.4ex}}

\newcommand{\N}{\mathbb{N}}

\newcommand{\R}{\mathbb{R}}

\newcommand{\A}{\mathcal{A}}
\newcommand{\B}{\mathcal{B}}

\newcommand{\M}{\mathcal{M}}
\renewcommand{\S}{\mathcal{S}}

\newcommand{\nN}{_{n \in \N}}

\newcommand{\prob}[1]{P_{#1}}
\newcommand{\D}{\mathcal{D}} % distributions

\newcommand{\norm}[1]{||#1||}
\newcommand{\mat}[1]{\M_{Q \times Q}(#1)}
\newcommand{\matS}[1]{\S_{Q \times Q}(#1)}
\newcommand{\tr}[1]{\langle #1 \rangle}

\newcommand{\inter}[1]{\overline{#1}}
\newcommand{\boo}[1]{\pi(#1)}

\renewcommand{\P}{\mathcal{P}}
\newcommand{\PA}{\mathcal{P}A^*}
\newcommand{\PAF}{\PA_f}

\newcommand{\ed}{\text{end}}

\newcommand{\chck}{\text{check}}

%% file: intro.tex
Rabin~\cite{Rabin63} introduced the notion of probabilistic
automata, which are finite automata with randomised transitions.
This powerful model has been widely studied since then and has applications,
for instance in image processing, computational biology and speech processing.
This paper follows a long line of work that studies the algorithmic properties of probabilistic automata.
We consider the value $1$ problem:
it asks, given a probabilistic automaton, whether there exist words accepted 
with probability arbitrarily close to $1$.

This problem has been shown undecidable~\cite{GimbertOualhadj10}.
Different approaches led to construct subclasses of probabilistic automata for which the value $1$ problem
is decidable; the first class was $\sharp$-acylic automata~\cite{GimbertOualhadj10},
then concurrently simple automata~\cite{ChatterjeeTracol12} and leaktight automata~\cite{FijalkowGimbertOualhadj12}.
It has been shown in~\cite{FijalkowGimbertKelmendiOualhadj15} that the so-called Markov Monoid algorithm introduced in~\cite{FijalkowGimbertOualhadj12}
is the most correct algorithm of the three algorithms.
Indeed, both $\sharp$-acylic and simple automata are strictly subsumed by leaktight automata,
for which the Markov Monoid algorithm correctly solves the value $1$ problem.

\bigskip
Yet we were missing a good understanding of the computations realised by the Markov Monoid algorithm.
The aim of this paper is to provide such an insight by giving a characterisation of this algebraic algorithm.
We show the existence of \textit{convergence speeds} phenomena, which can be polynomial or exponential.
Our main technical contribution is to prove that the Markov Monoid algorithm captures exactly \textit{polynomial behaviours}.

\bigskip
Proving this characterisation amounts to giving precise bounds on convergences of non-homogeneous Markov chains.
Our second contribution is to define a new framework allowing us to rephrase this characterisation and to give a modular proof for it,
using techniques from topology and linear algebra.
We develop a profinite approach for probabilistic automata, called prostochastic theory.
This is inspired by the profinite approach for (classical) automata~\cite{Almeida05,Pin09,GehrkeGrigorieffPin10},
and for automata with counters~\cite{Torunczyk11}.

\bigskip
Section~\ref{sec:mma} is devoted to defining the Markov Monoid algorithm and stating the characterisation:
it answers ``YES'' if, and only if, the probabilistic automaton accepts some polynomial sequence.

In Section~\ref{sec:prostochastic}, we introduce a new framework, the prostochastic theory, which is used to restate and prove the characterisation.
We first construct a space called the free prostochastic monoid, whose elements are called prostochastic words.
We define the acceptance of a prostochastic word by a probabilistic automaton, 
and show that the value $1$ problem can be reformulated as the emptiness problem for probabilistic automata over prostochastic words.
We then explain how to construct non-trivial prostochastic words, by defining a limit operator $\omega$, leading to the definition of polynomial prostochastic words.
The above characterisation above reads in the realm of prostochastic theory as follows:
the Markov Monoid algorithm answers ``YES'' if, and only if, the probabilistic automaton accepts some polynomial prostochastic word.

Section~\ref{sec:conclusions} concludes by showing how this characterisation, 
combined with an improved undecidability result, 
supports the claim that the Markov Monoid algorithm is \textit{in some sense} optimal.

%% file: defs.tex
Let $Q$ be a finite set of states.

A (probability) distribution over $Q$ is a function $\delta : Q \to [0,1]$ such that $\sum_{q \in Q} \delta(q) = 1$.
We denote $\D(Q)$ the set of distributions over $Q$, which we often consider as vectors indexed by $Q$.
%We denote by $\frac{1}{3} \cdot q + \frac{2}{3} \cdot q'$ the distribution
%that picks $q$ with probability $\frac{1}{3}$ and $q'$ with probability $\frac{2}{3}$,
%and by $q$ the trivial distribution picking $q$ with probability $1$.

For $E \subseteq \R$, we denote $\mat{E}$ the set of (square) matrices indexed by $Q$ over $E$.
We denote $I$ the identity matrix.
A matrix $M \in \mat{\R}$ is stochastic if each row is a distribution over $Q$; 
the subset consisting of stochastic matrices is denoted $\matS{E}$.
The space $\matS{\R}$ is equipped with the norm $\norm{\cdot}$ defined by
$$\norm{M} = \max_{s \in Q} \sum_{t \in Q} |M(s,t)|.$$
This induces the standard Euclidean topology on $\matS{\R}$.
The following classical properties will be useful:

\begin{fact}{(Topology of the stochastic matrices)}
\begin{itemize}
	\item For all matrix $M \in \matS{\R}$, we have $\norm{M} = 1$,
	\item For all matrices $M,M' \in \mat{\R}$, we have $\norm{M \cdot M'} \le \norm{M} \cdot \norm{M'}$,
	\item The monoid $\matS{\R}$ is (Hausdorff) compact.
\end{itemize}
\end{fact}

\begin{definition}{(Probabilistic automaton)}
A \textit{probabilistic automaton} $\A$ is given by a finite set of states $Q$,
a transition function $\phi : A \to \matS{\R}$,
an initial state $q_0 \in Q$
and a set of final states $F \subseteq Q$.
\end{definition}

%This simple and natural definition was given by Rabin in~\cite{Rabin63}.
Observe that it generalises the definition for classical deterministic automata,
in which transitions functions are $\phi : A \to \matS{\set{0,1}}$.
We allow here the transition functions of probabilistic automata to have arbitrary real values;
when considering computational properties, we assume that they are rational numbers.

A transition function $\phi : A \to \matS{\R}$ naturally
induces a morphism $\phi : A^* \to \matS{\R}$\footnote{Note that we use ``morphism'' for ``monoid homomorphism'' throughout the paper.}.
%We sometimes give $\phi$ as a function $Q \times A \to \D(Q)$, which is equivalent.

We denote $\prob{\A}(s \xrightarrow{w} t)$ 
the probability to go from state $s$ to state $t$ reading $w$ 
on the automaton $\A$,
\textit{i.e.} $\phi(w)(s,t)$.
We extend the notation: for a subset $T$ of the set of states,
$\prob{\A}(s \xrightarrow{w} T)$ is defined by $\phi(w)(s,T) = \sum_{t \in T} \phi(w)(s,t)$.
The \emph{acceptance probability} of a word $w \in A^*$ by $\A$ is 
$\prob{\A}(q_0 \xrightarrow{w} F)$, denoted $\prob{\A}(w)$.
In words, the above is the probability that a run starting from the initial state $q_0$
ends in a final state (\textit{i.e.} a state in $F$).
The \emph{value} of a probabilistic automaton $\A$ is $\val{\A} = \sup \set{\prob{\A}(w) \mid w \in A^*}$,
the supremum over all words of the acceptance probability.

\begin{definition}{(Value $1$ Problem)}
The value $1$ problem is the following decision problem: 
given a probabilistic automaton $\A$ as input, determine whether $\val{\A} = 1$,
\textit{i.e.} whether there exist words whose acceptance probability is arbitrarily close to $1$.
\end{definition}

Equivalently, the value $1$ problem asks for the existence of a sequence of words $(u_n)\nN$ such that $\lim_n \prob{\A}(u_n) = 1$.

%% file: mma.tex
The Markov Monoid algorithm was introduced in~\cite{FijalkowGimbertOualhadj12}, we give here a different yet equivalent presentation.
Consider a probabilistic automaton $\A$,
the Markov Monoid algorithm consists in computing, by a saturation process, the Markov Monoid of $\A$.

It is a monoid of Boolean matrices: 
all numerical values are projected to Boolean values.
So instead of considering $M \in \matS{\R}$, we are interested in $\boo{M} \in \mat{\set{0,1}}$,
the Boolean matrix such that $\boo{M}(s,t) = 1$ if $M(s,t) > 0$, and $\boo{M}(s,t) = 0$ otherwise.
Hence to define the Markov Monoid,
one can consider the underlying non-deterministic automaton $\boo{\A}$
instead of the probabilistic automaton $\A$.
Formally, $\boo{\A}$ is defined as $\A$, except that its transitions are given by $\boo{\phi(a)}$ for the letter $a \in A$.

The Markov Monoid of $\boo{\A}$ contains the transition monoid of $\boo{\A}$,
which is the monoid of Boolean matrix generated by $\set{\boo{\phi(a)} \mid a \in A}$.
Informally speaking, the transition monoid accounts for 
the Boolean action of every finite word.
Formally, for a word $w \in A^*$, the element $\tr{w}$ of 
the transition monoid of $\boo{\A}$
satisfies the following: $\tr{w}(s,t) = 1$ if, and only if, there exists
a run from $s$ to $t$ reading $w$ on $\boo{\A}$.

The Markov Monoid extends the transition monoid
by introducing a new operator, the stabilisation.
On the intuitive level first: let $M \in \matS{\R}$, 
it can be interpreted as a Markov chain;
its Boolean projection $\boo{M}$ represents the structural properties of this Markov chain.
The stabilisation $\boo{M}^\sharp$ accounts for $\lim_n M^n$,
\textit{i.e.} the behaviour of the Markov chain $M$ in the limit.

To give the formal definition of the stabilisation operator, we need a few more notations.
As a convention, $M$ denotes a matrix in $\matS{\R}$, and $m$ a Boolean matrix.
Note that when considering stochastic matrices we compute in the real semiring,
and when considering Boolean matrices, we compute products in the Boolean semiring,
leading to two distinct notions of idempotent matrices.

The following definitions mimick the notions of recurrent and transient states from Markov chain theory.

\begin{definition}{(Idempotent Boolean matrix, recurrent and transient state)}
Let $m$ be a Boolean matrix.
It is idempotent if $m \cdot m = m$.

Assume $m$ is idempotent. We say that:
\begin{itemize}
	\item the state $s \in Q$ is $m$-recurrent if for all $t \in Q$, if $m(s,t) = 1$, then $m(t,s) = 1$,
and it is $m$-transient if it is not $m$-recurrent,

	\item the $m$-recurrent states $s,t \in Q$ belong to the same recurrence class if $m(s,t) = 1$.
\end{itemize}
\end{definition}

\begin{definition}{(Stabilisation)}
Let $m$ be a Boolean idempotent matrix.

The stabilisation of $m$ is denoted $m^\sharp$ and defined by:
$$m^\sharp(s,t) = 
\begin{cases}
1 & \textrm{if } m(s,t) = 1 \textrm{ and } t \textrm{ is } m\textrm{-recurrent,} \\
0 & \textrm{otherwise.}
\end{cases}$$
\end{definition}

The definition of the stabilisation
matches the intuition that
in the Markov chain $\lim_n M^n$,
the probability to be in non-recurrent states converges to $0$.
%This will be made precise in Subsection~\ref{subsec:mma:consistency}.

\begin{definition}{(Markov Monoid)}
The Markov Monoid of an automaton $\A$ is the smallest set of Boolean matrices
containing $\set{\boo{\phi(a)} \mid a \in A}$ closed under product and stabilisation of idempotents.
\end{definition}

\begin{algorithm}[ht]
\caption{The Markov Monoid algorithm.}
\label{algo:markov_monoid}
\SetAlgoLined
\KwData{A probabilistic automaton.}
%\KwResult{Correct answer to the value $1$ problem.}
     
$\M \gets \set{\boo{\phi(a)} \mid a \in A} \cup \set{I}$.

\Repeat{there is nothing to add}{
	\If{there is $m,m' \in \M$ such that $m \cdot m' \notin \M$}{
		add $m \cdot m'$ to $\M$
		}

	\If{there is $m \in \M$ such that $m$ is idempotent and $m^\sharp \notin \M$}{
		add $m^\sharp$ to $\M$
		}
}

\eIf{there is a value $1$ witness in $\M$}{
	return YES\;}{
	return NO\;}{
}
\end{algorithm}

On an intuitive level, a Boolean matrix in the Markov Monoid reflects
the asymptotic behaviour of a sequence of finite words.

The Markov Monoid algorithm computes the Markov Monoid,
and looks for \textit{value $1$ witnesses}:
%, detailed in Algorithm~\ref{algo:markov_monoid} 

\begin{definition}{(Value 1 witness)}
Let $\A$ be a probabilistic automaton.

A Boolean matrix $m$ is a value $1$ witness if:
for all states $t \in Q$, if $m(q_0,t) = 1$, then $t \in F$.
\end{definition}

The Markov Monoid algorithm answers ``YES'' 
if there exists a value $1$ witness in the Markov Monoid,
and ``NO'' otherwise.

\bigskip
Our main technical result is the following theorem, which is a characterisation of the Markov Monoid algorithm.
It relies on the notion of \textit{polynomial sequences of words}.

We define two operations for sequences of words, mimicking the operations of the Markov Monoid.
\begin{itemize}
	\item the first is concatenation: given $(u_n)\nN$ and $(v_n)\nN$, the concatenation is the sequence $(u_n \cdot v_n)\nN$,
	\item the second is iteration: given $(u_n)\nN$, its iteration is the sequence $(u_n^n)\nN$;
	the $n$\textsuperscript{th} word is repeated $n$ times.
\end{itemize}

\begin{definition}{(Polynomial sequence)}
The class of polynomial sequences is the smallest class of sequences containing the constant sequences $(a)\nN$ for each letter $a \in A$
and $(\varepsilon)\nN$, closed under concatenation and iteration.
\end{definition}

A typical example of a polynomial sequence is $((a^n b)^n)\nN$,
and a typical example of a sequence which is not polynomial is $\left((a^n b)^{2^n}\right)\nN$.

We proceed to our main result:

\begin{theorem}{(Characterisation of the Markov Monoid algorithm)}
\label{thm:characterisation_mma}
The Markov Monoid algorithm answers ``YES'' on input $\A$
if, and only if,
there exists a polynomial sequence $(u_n)\nN$ such that $\lim_n \prob{\A}(u_n) = 1$.
\end{theorem}

This result could be proved directly, without appealing to the prostochastic theory developed in the next section.
The proof relies on technically intricate calculations over non-homogeneous Markov chains;
the prostochastic theory allows to simplify its presentation, making it more modular.
We will give the proof of Theorem~\ref{thm:characterisation_mma} in Subsection~\ref{subsec:prostochastic:reformulation_characterisation},
after restating it using the prostochastic theory.

A second advantage of using the prostochastic theory is to give a more natural and robust definition of polynomial sequences,
which in the prostochastic theory correspond to polynomial prostochastic words.

A direct corollary of Theorem~\ref{thm:characterisation_mma} is the absence of false negatives:

\begin{corollary}{(No false negatives for the Markov Monoid algorithm)}

If the Markov Monoid algorithm answers ``YES'' on input $\A$,
then $\A$ has value $1$.
\end{corollary}

%% file: theory.tex
%The set $\mat{\R}$ is equipped with the norm $\norm{\cdot}$ defined by
%$\norm{M} = \max_{s \in Q} \sum_{t \in Q} |M(s,t)|$.
%The following classical properties will be useful:
%
%\begin{fact}{(Topology of the stochastic matrices)}
%\begin{itemize}
%	\item For every matrix $M \in \matS{\R}$, we have $\norm{M} = 1$,
%	\item For every matrices $M,M' \in \mat{\R}$, we have $\norm{M \cdot M'} \le \norm{M} \cdot \norm{M'}$,
%	\item The monoid $\matS{\R}$ is compact.
%\end{itemize}
%\end{fact}

The purpose of the prostochastic theory is to construct a (Hausdorff) compact\footnote{Following the French tradition, here by compact we mean Hausdorff compact: distinct points have disjoint neighbourhoods.} monoid $\PA$ together with an injective morphism $\iota : A^* \to \PA$,
called the free prostochastic monoid,
satisfying the following universal property:
\begin{center}
``Every morphism $\phi : A^* \to \matS{\R}$ extends uniquely\\
to a continuous morphism $\widehat{\phi} : \PA \to \matS{\R}$.''
\end{center}
Here, by ``$\widehat{\phi}$ extends $\phi$'' we mean $\phi = \widehat{\phi} \circ \iota$.

We give two statements about $\PA$, the first will be weaker but enough for our purposes in this paper,
and the second more precise, and justifying the name ``free prostochastic monoid''. 
The reason for giving two statements is that the first avoids a number of technical points
that will not play any further role, so the reader interested in the applications to the Markov Monoid algorithm
may skip this second statement.

\begin{theorem}{(Existence of the free prostochastic monoid -- weaker statement)}
\label{thm:free_prostochastic_monoid}
For every finite alphabet $A$,
there exists a compact monoid $\PA$ and an injective morphism $\iota : A^* \to \PA$
such that every morphism $\phi : A^* \to \matS{\R}$ extends uniquely 
to a continuous morphism $\widehat{\phi} : \PA \to \matS{\R}$.
\end{theorem}

We construct $\PA$ and $\iota$.
Consider $X = \prod_{\phi : A^* \to \matS{\R}} \matS{\R}$, the product of several copies of $\matS{\R}$,
one for each morphism $\phi : A^* \to \matS{\R}$.
An element $m$ of $X$ is denoted $(m(\phi))_{\phi : A^* \to \matS{\R}}$: it is given by an element $m(\phi)$ of $\matS{\R}$ 
for each morphism $\phi : A^* \to \matS{\R}$.
Thanks to Tychonoff's theorem, the monoid $X$ equipped with the product topology is compact.
%We define the product norm on $X$ by $\norm{m} = \sup \set{\norm{m(\phi)} \mid \phi : A^* \to \matS{\R}}$,
%inducing the product topology on $X$.

Consider the map $\iota : A \to X$ defined by $\iota(a) = (\phi(a))_{\phi : A \to \P}$, it induces
an injective morphism $\iota : A^* \to X$.
To simplify notation, we sometimes assume that $A \subseteq X$ and denote $a$ for $\iota(a)$.

Denote $\PA = \overline{A^*}$, the closure of $A^*$ in $X$. 
Note that it is a compact monoid: the compactness follows from the fact that it is closed in $X$.
By definition, an element $\overline{u}$ of $\PA$, called a \textit{prostochastic word}, 
is obtained as the limit in $\PA$ of a sequence $\mathbf{u}$ of finite words.
In this case we write $\lim \mathbf{u} = \overline{u}$ and say that $\mathbf{u}$ induces $\overline{u}$.

Note that by definition of the product topology on $X$,
a sequence of finite words $\mathbf{u}$ converges in $X$ if, and only if, 
for all morphisms $\phi : A^* \to \matS{\R}$, the sequence of stochastic matrices $\phi(\mathbf{u})$ converges.

We say that two converging sequences of finite words $\mathbf{u}$ and $\mathbf{v}$ are equivalent
if they induce the same prostochastic word, \textit{i.e.} if $\lim \mathbf{u} = \lim \mathbf{v}$.
Equivalently, two converging sequences of finite words $\mathbf{u}$ and $\mathbf{v}$ are equivalent
if, and only if, 
for all morphisms $\phi : A^* \to \matS{\R}$, we have $\lim \phi(\mathbf{u}) = \lim \phi(\mathbf{v})$.

\begin{proof}
We prove that $\PA$ satisfies the universal property.
Consider a morphism $\phi : A^* \to \matS{\R}$, and define $\widehat{\phi} : \PA \to \matS{\R}$ 
by $\widehat{\phi}(\overline{u}) = \lim \phi(\mathbf{u})$,
where $\mathbf{u}$ is \textit{some} sequence of finite words inducing $\overline{u}$.
This is well defined and extends $\phi$.
Indeed, consider two equivalent sequences of finite words $\mathbf{u}$ and $\mathbf{v}$ inducing $\overline{u}$.
By definition, for all $\psi : A^* \to \matS{\R}$, we have $\lim \psi(\mathbf{u}) = \lim \psi(\mathbf{v})$,
so in particular for $\phi$ this implies $\lim \phi(\mathbf{u}) = \lim \phi(\mathbf{v})$,
and $\widehat{\phi}$ is well defined.
Both continuity and uniqueness are clear.

We prove that $\widehat{\phi}$ is a morphism.
Consider 
$$D = \set{(\overline{u},\overline{v}) \in \PA \times \PA \mid 
\widehat{\phi}(\overline{u} \cdot \overline{v}) = \widehat{\phi}(\overline{u}) \cdot \widehat{\phi}(\overline{v})}.$$
To prove that $\widehat{\phi}$ is a morphism, we prove that $D = \PA \times \PA$.
First of all, $A^* \times A^* \subseteq D$.
Since $A^* \times A^*$ is dense in $\PA \times \PA$, it suffices to show that $D$ is closed.
This follows from the continuity of both product functions in $\PA$ and in $\matS{\R}$ as well as of $\widehat{\phi}$.
\hfill\qed\end{proof}

\bigskip
We give a second, stronger statement about $\PA$,
which in particular justifies the name ``free prostochastic monoid''.

From now on, by ``monoid'' we mean ``compact topological monoids''.
The term topological means that the product function is continuous:
$$\begin{array}{ccc}
\P \times \P & \to & \P \\
(s,t) & \mapsto & s \cdot t
\end{array}$$

A monoid is profinite if any two elements can be distinguished
by a continuous morphism into a finite monoid, \textit{i.e.} by a finite automaton.
(Formally speaking, this is the definition of residually finite monoids, 
which coincide with profinite monoids for compact monoids, see~\cite{Almeida05}.)
To define prostochastic monoids, we use a stronger distinguishing feature, namely probabilistic automata.

\begin{definition}{(Prostochastic monoid)}
A monoid $\P$ is prostochastic if for all elements $s \neq t$ in $\P$,
there exists a continuous morphism $\psi : \P \to \matS{\R}$ such that $\psi(s) \neq \psi(t)$.
\end{definition}

There are many more prostochastic monoids than profinite monoids. 
Indeed, $\matS{\R}$ is prostochastic, but not profinite in general.

The following theorem extends Theorem~\ref{thm:free_prostochastic_monoid}.
The statement is the same as in the profinite theory,
replacing ``profinite monoid'' by ``prostochastic monoid''.

\begin{theorem}{(Existence of the free prostochastic monoid -- stronger statement)}
\label{thm:free_prostochastic_monoid_strong}
For every finite alphabet $A$,
\begin{enumerate}
	\item There exists a prostochastic monoid $\PA$ and an injective morphism $\iota : A^* \to \PA$
	such that every morphism $\phi : A^* \to \P$, 
	where $\P$ is a prostochastic monoid, extends uniquely 
	to a continuous morphism $\widehat{\phi} : \PA \to \P$.

	\item All prostochastic monoids satisfying this universal property are
	homeomorphic.
\end{enumerate}
The unique prostochastic monoid satisfying the universal property stated 
in item 1. is called the free prostochastic monoid, and denoted $\PA$.
\end{theorem}

\begin{proof}
We prove that $\PA$ satisfies the stronger universal property, along the same lines as for the weaker one.
Consider a morphism $\phi : A^* \to \P$, and define $\widehat{\phi} : \PA \to \P$ 
by $\widehat{\phi}(\overline{u}) = \lim \phi(\mathbf{u})$,
where $\mathbf{u}$ is \textit{some} sequence of finite words inducing $\overline{u}$.

To see that this is well defined, we use the fact that $\P$ is prostochastic.
Consider two equivalent sequences of finite words $\mathbf{u}$ and $\mathbf{v}$ inducing $\overline{u}$.
Consider a continuous morphism $\psi : \P \to \matS{\R}$, the composition $\psi \circ \phi$ is a continuous morphism
from $A^*$ to $\matS{\R}$, so since $\mathbf{u}$ and $\mathbf{v}$ are equivalent it follows that 
$\lim (\psi \circ \phi)(\mathbf{u}) = \lim (\psi \circ \phi)(\mathbf{v})$,
\textit{i.e.} $\lim \psi (\phi(\mathbf{u})) = \lim \psi (\phi(\mathbf{v}))$.
Since $\psi$ is continuous, this implies $\psi (\lim \phi(\mathbf{u})) = \psi (\lim \phi(\mathbf{v}))$.
We proved that for all continuous morphisms $\psi : \P \to \matS{\R}$, we have 
$\psi (\lim \phi(\mathbf{u})) = \psi (\lim \phi(\mathbf{v}))$;
since $\P$ is prostochastic, it follows that $\lim \phi(\mathbf{u}) = \lim \phi(\mathbf{v})$,
and $\widehat{\phi}$ is well defined.

Clearly $\widehat{\phi}$ extends $\phi$. Both continuity and uniqueness are clear.
We prove that $\widehat{\phi}$ is a morphism.
Consider 
$$D = \set{(\overline{u},\overline{v}) \in \PA \times \PA \mid 
\widehat{\phi}(\overline{u} \cdot \overline{v}) = \widehat{\phi}(\overline{u}) \cdot \widehat{\phi}(\overline{v})}.$$
To prove that $\widehat{\phi}$ is a morphism, we prove that $D = \PA \times \PA$.
First of all, $A^* \times A^* \subseteq D$.
Since $A^* \times A^*$ is dense in $\PA \times \PA$, it suffices to show that $D$ is closed.
This follows from the continuity of both product functions in $\PA$ and in $\P$ as well as of $\widehat{\phi}$.

\bigskip
We prove that $\PA$ is prostochastic.
Let $\overline{u} \neq \overline{v}$ in $\PA$.
Consider two sequences of finite words $\mathbf{u}$ and $\mathbf{v}$ inducing respectively $\overline{u}$ and $\overline{v}$,
there exists a morphism $\phi : A^* \to \matS{\R}$ such that $\lim \phi(\mathbf{u}) \neq \lim \phi(\mathbf{v})$.
Thanks to the universal property proved in the first point,
this induces a continuous morphism $\widehat{\phi} : \PA \to \matS{\R}$
such that $\widehat{\phi}(\mathbf{u}) \neq \widehat{\phi}(\mathbf{v})$, finishing the proof that $\PA$ is prostochastic.

\bigskip
We now prove that there is a unique prostochastic monoid satisfying the universal property, up to homeomorphism.
Let $\P_1$ and $\P_2$ be two prostochastic monoids satisfying the universal property, together with two injective morphisms 
$\iota_1 : A^* \to \P_1$ and $\iota_2 : A^* \to \P_2$.
Thanks to the universal property, $\iota_1$ and $\iota_2$ are extended to continuous morphisms 
$\widehat{\iota_1} : \P_2 \to \P_1$ and $\widehat{\iota_2} : \P_1 \to \P_2$,
and $\widehat{\iota_1} \circ \iota_2 = \iota_1$ and $\widehat{\iota_2} \circ \iota_1 = \iota_2$.
This implies that $\widehat{\iota_1} \circ \widehat{\iota_2} \circ \iota_1 = \iota_1$;
thanks to the universal property again, there exists a unique continuous morphism $\theta$ such that $\theta \circ \iota_1 = \iota_1$,
and since both $\widehat{\iota_1} \circ \widehat{\iota_2}$ and the identity morphism on $\P_1$ satisfy this equality,
it follows that they are equal. 
Similarly, $\widehat{\iota_2} \circ \widehat{\iota_1}$ is equal to the identity morphism on $\P_2$.
It follows that $\widehat{\iota_1}$ and $\widehat{\iota_2}$ are mutually inverse homeomorphisms between $\P_1$ and $\P_2$.
\hfill\qed\end{proof}

\begin{remark}
We remark that the free prostochastic monoid $\PA$ contains the free profinite monoid $\widehat{A^*}$.
To see this, we start by recalling some properties of $\widehat{A^*}$,
which is the set of \textit{converging} sequences up to \textit{equivalence}, where:
\begin{itemize}
	\item a sequence of finite words $\mathbf{u}$ is converging if, and only if, for every deterministic automaton $\A$,
	the sequence is either ultimately accepted by $\A$ or ultimately rejected by $\A$,
	\textit{i.e.}
	there exists $N \in \N$ such that either for all $n \ge N$, the word $u_n$ is accepted by $\A$,
	or for all $n \ge N$, the word $u_n$ is rejected by $\A$,
	\item two sequences of finite words $\mathbf{u}$ and $\mathbf{v}$ are equivalent if for every deterministic automaton $\A$,
	either both sequences are ultimately accepted by $\A$, or both sequences are ultimately rejected by $\A$.
\end{itemize}
Clearly:
\begin{itemize}
	\item if a sequence of finite words is converging with respect to $\PA$, then it is converging with respect to $\widehat{A^*}$,
	as deterministic automata form a subclass of probabilistic automata,
	\item if two sequences of finite words are equivalent with respect to $\PA$, then they are equivalent with respect to $\widehat{A^*}$.
\end{itemize}
Each profinite word induces at least one prostochastic word: 
by compactness of $\PA$, each sequence of finite words $\mathbf{u}$
contains a converging subsequence with respect to $\PA$.
This defines an injection from $\widehat{A^*}$ into $\PA$.
In particular, this implies that $\PA$ is uncountable. 
Since it is defined as the topological of a countable set, it has the cardinality of the continuum.
\end{remark}

%% file: reformulation.tex
The aim of this subsection is to show that the value $1$ problem, which talks about sequences of finite words,
can be reformualted as an emptiness problem over prostochastic words.

\begin{definition}{(Prostochastic language of a probabilistic automaton)}
Let $\A$ be a probabilistic automaton, $\phi$ is the transition function of $\A$.
The prostochastic language of $\A$ is:
$$L(\A) = \set{\overline{u} \mid \widehat{\phi}(\overline{u})(q_0,F) = 1}.$$ 
We say that $\A$ accepts a prostochastic word $\overline{u}$ if $\overline{u} \in L(\A)$.
\end{definition}

\begin{theorem}{(Reformulation of the value 1 problem)}
\label{thm:equivalence_prostochastic}
Let $\A$ be a probabilistic automaton.
The following are equivalent:
\begin{itemize}
	\item $\val{\A} = 1$,
	\item $L(\A)$ is non-empty.
\end{itemize}
\end{theorem}

\begin{proof}
Assume $\val{\A} = 1$, then there exists a sequence of words $\mathbf{u}$ such that $\lim \prob{\A}(\mathbf{u}) = 1$.
We see $\mathbf{u}$ as a sequence of prostochastic words.
By compactness of $\PA$ it contains a converging subsequence.
The prostochastic word induced by this subsequence belongs to $L(\A)$.

Conversely, let $\overline{u}$ in $L(\A)$, \textit{i.e.} such that $\widehat{\phi}(\overline{u})(q_0,F) = 1$.
Consider a sequence of finite words $\mathbf{u}$ inducing $\overline{u}$.
By definition, we have $\lim \phi(\mathbf{u})(q_0,F) = 1$,
\textit{i.e.} $\lim \prob{\A}(\mathbf{u}) = 1$, implying that $\val{\A} = 1$.
\hfill\qed\end{proof}

%% file: limit.tex
We show in this subsection how to construct non-trivial prostochastic words, and in particular the polynomial prostochastic words.
To this end, we need to better understand \textit{convergence speeds phenomena}:
different limit behaviours can occur, depending on how fast the underlying Markov chains converge.

\bigskip
We define a limit operator $\omega$.
Consider the function $f : \N \to \N$ defined by $f(n) = k!$, where $k$ is maximal such that $k! \le n$.
The function $f$ grows linearly; the choice of $n$ is arbitrary, one could replace $n$ by any polynomial, or even by any subexponential function,
see Remark~\ref{rem:polynomial_prostochastic_words}.

The operator $\omega$ takes as input a sequence of finite words, and outputs a sequence of finite words.
Formally, let $\mathbf{u}$ be a sequence of finite words, define:
$$\mathbf{u}^\omega = (u_n^{f(n)})\nN.$$

It is not true in general that if $\mathbf{u}$ converges, then $\mathbf{u}^\omega$ converges.
We will show that a sufficient condition is that $\mathbf{u}$ is fast.

We say that a sequence $(M_n)\nN$ converges exponentially fast to $M$ if there exists a constant $C > 1$ 
such that for all $n$ large enough, $\norm{M_n - M} \le C^{-n}$.

\begin{definition}{(Fast sequence)}
A sequence of finite words $\mathbf{u}$ is fast if it converges
(we denote $\overline{u}$ the prostochastic word it induces),
and for every morphism $\phi : A^* \to \matS{\R}$, 
the sequence $(\phi(u_n))\nN$ converges exponentially fast.
\end{definition}

A prostochastic word is \textit{fast} if it is induced by \textit{some} fast sequence.
We denote $\PAF$ the set of fast prostochastic words.
Note that a priori, not all prostochastic words are induced by some fast sequence.

\bigskip
We first prove that $\PAF$ is a submonoid of $\PA$.

\begin{lemma}{(The concatenation of two fast sequences is fast)}
\label{lem:concatenation_fast}
Let $\mathbf{u},\mathbf{v}$ be two fast sequences.

The sequence $\mathbf{u \cdot v} = (u_n \cdot v_n)\nN$ is fast.
\end{lemma}

\begin{proof}
Consider a morphism $\phi : A^* \to \matS{\R}$ and $n \in \N$.
\begin{eqnarray*}
\lefteqn{\norm{\phi(u_n \cdot v_n) - \widehat{\phi}(\overline{u} \cdot \overline{v})}} \\
& = & \norm{\phi(u_n) \cdot \phi(v_n) - \widehat{\phi}(\overline{u}) \cdot \widehat{\phi}(\overline{v})} \\
& = & \norm{\phi(u_n) \cdot (\phi(v_n) - \widehat{\phi}(\overline{v})) - (\widehat{\phi}(\overline{u}) - \phi(u_n)) \cdot \widehat{\phi}(\overline{v})}\\ 
& \le & \norm{\phi(u_n)} \cdot \norm{\phi(v_n) - \widehat{\phi}(\overline{v})} + \norm{\widehat{\phi}(\overline{u}) - \phi(u_n)} \cdot \norm{\widehat{\phi}(\overline{v})}\\ 
& = & \norm{\phi(v_n) - \widehat{\phi}(\overline{v})} + \norm{\widehat{\phi}(\overline{u}) - \phi(u_n)}.
\end{eqnarray*}
Since $\mathbf{u}$ and $\mathbf{v}$ are fast, the previous inequality implies that $\mathbf{u \cdot v}$ is fast.
\hfill\qed\end{proof}

Let $\overline{u}$ and $\overline{v}$ be two fast prostochastic words, thanks to Lemma~\ref{lem:concatenation_fast},
the prostochastic word $\overline{u} \cdot \overline{v}$ is fast.

\bigskip
The remainder of this subsection is devoted to proving that $\omega$ is an operator $\PAF \to \PAF$.
This is the key technical point of our characterisation. 
Indeed, we will define polynomial prostochastic words using concatenation and the operator $\omega$,
mimicking the definition of polynomial sequences of finite words.
The fact that $\omega$ preserves the fast property of prostochastic words allows 
to obtain a perfect correspondence between polynomial sequences of finite words and polynomial prostochastic words.

The main technical tool is the following theorem, stating the exponentially fast convergence of the powers of a stochastic matrix.

\begin{theorem}{(Powers of a stochastic matrix)}
\label{thm:powers_stochastic_matrix}
Let $M \in \matS{\R}$. Denote $P = M^{|Q|!}$.
Then the sequence $(P^n)_{n \in \N}$ converges exponentially fast to a matrix $M^\omega$, satisfying:
$$\boo{M^\omega}(s,t) = 
\begin{cases}
1 & \textrm{if } \boo{P}(s,t) = 1 \textrm{ and } t \textrm{ is } \boo{P}\textrm{-recurrent,} \\
0 & \textrm{otherwise.}
\end{cases}$$
\end{theorem}

The proof of Theorem~\ref{thm:powers_stochastic_matrix} is given in Subsection~\ref{subsec:prostochastic:proof}.

The following lemma shows that the $\omega$ operator is well defined for fast sequences.
The second item shows that $\omega$ commutes with morphisms.

\begin{lemma}{(Limit operator for fast sequences)}
\label{lem:limit_fast}
Let $\mathbf{u}, \mathbf{v}$ be two equivalent fast sequences, inducing the fast prostochastic word $\overline{u}$.
Then the sequences $\mathbf{u}^\omega$ and $\mathbf{v}^\omega$ are fast and equivalent,
inducing the fast prostochastic word denoted $\overline{u}^\omega$.

Furthermore, for every morphism $\phi : A^* \to \matS{\R}$, we have $\widehat{\phi}(\overline{u}^\omega) = \widehat{\phi}(\overline{u})^\omega$.
\end{lemma}

\begin{proof}
Let $\phi : A \to \matS{\R}$.

The sequence $(\widehat{\phi}(\overline{u})^{f(n)})\nN$ is a subsequence of $(\widehat{\phi}(\overline{u})^{|Q|! \cdot n})\nN$,
so Theorem~\ref{thm:powers_stochastic_matrix} implies that it converges exponentially fast to $\widehat{\phi}(\overline{u})^\omega$.
It follows that there exists a constant $C_1 > 1$ such that for all $n$ large enough,
we have $\norm{\widehat{\phi}(\overline{u})^{f(n)} - \widehat{\phi}(\overline{u})^\omega} \le C_1^{-f(n)}$.

We proceed in two steps, using the following inequality, which holds for every $n$:
$$\norm{\phi(u_n^{f(n)}) - \widehat{\phi}(\overline{u})^\omega}
\le 
\norm{\phi(u_n)^{f(n)} - \widehat{\phi}(\overline{u})^{f(n)}} 
+ 
\norm{\widehat{\phi}(\overline{u})^{f(n)} - \widehat{\phi}(\overline{u})^\omega}.$$

For the left summand, we rely on the following equality, where $x$ and $y$ may not commute:
$$x^N - y^N = \sum_{k = 0}^{N-1} x^{N-k-1} \cdot (x - y) \cdot y^k.$$
Leting $N = f(n)$, this gives:
\begin{eqnarray*}
\lefteqn{\norm{\phi(u_n)^N - \widehat{\phi}(\overline{u})^N}} \\ 
&& =
\norm{\sum_{k = 0}^{N-1} \phi(u_n)^{N-k-1} \cdot (\phi(u_n) - \widehat{\phi}(\overline{u})) \cdot \widehat{\phi}(\overline{u})^k} \\ 
&& \le 
\sum_{k = 0}^{N-1} 
\norm{\phi(u_n)^{N-k-1}} 
\cdot \norm{\phi(u_n) - \widehat{\phi}(\overline{u})} 
\cdot \norm{\widehat{\phi}(\overline{u})^k}  \\
&& \le  
\sum_{k = 0}^{N-1} 
\underbrace{\norm{\phi(u_n)}^{N-k-1}}_{= 1}
\cdot \norm{\phi(u_n) - \widehat{\phi}(\overline{u})} 
\cdot \underbrace{\norm{\widehat{\phi}(\overline{u})}^k}_{= 1}  \\
&& =  
N \cdot \norm{\phi(u_n) - \widehat{\phi}(\overline{u})}.
\end{eqnarray*}
Since $\mathbf{u}$ is fast, there exists a constant $C_2 > 1$ such that $\norm{\phi(u_n) - \widehat{\phi}(\overline{u})} \le C_2^{-n}$.
Altogether, we have 
$$\norm{\phi(u_n^{f(n)}) - \widehat{\phi}(\overline{u})^\omega} \le f(n) \cdot C_2^{-n} + C_1^{-f(n)}.$$
To conclude, observe that for all $n$ large enough, we have $\frac{n}{\log(n)} \le f(n) \le n$.
It follows that the sequence $\mathbf{u}^\omega$ is fast, and that
$\phi(\mathbf{u}^\omega)$ converges to $\widehat{\phi}(\overline{u})^\omega$.

\medskip
Furthermore, since $\mathbf{u}$ and $\mathbf{v}$ are equivalent,
we have $\lim \phi(\mathbf{u}) = \lim \phi(\mathbf{v})$,
\textit{i.e.} $\widehat{\phi}(\overline{u}) = \widehat{\phi}(\overline{v})$,
so $\widehat{\phi}(\overline{u})^\omega = \widehat{\phi}(\overline{v})^\omega$,
\textit{i.e.} $\lim \phi(\mathbf{u}^\omega) = \lim \phi(\mathbf{v}^\omega)$,
This implies that $\mathbf{u}^\omega$ and $\mathbf{v}^\omega$ are equivalent.
\hfill\qed\end{proof}

Let $\overline{u}$ be a fast prostochastic word, we define the prostochastic word $\overline{u}^\omega$
as induced by $\mathbf{u}^\omega$, for some sequence $\mathbf{u}$ inducing $\overline{u}$.
Thanks to Lemma~\ref{lem:limit_fast}, the prostochastic word $\overline{u}^\omega$ is well defined,
and fast.

We can now define polynomial prostochastic words.

First, $\omega$-expressions are described by the following grammar:
$$E \qquad \longrightarrow \qquad a \quad \mid \quad E \cdot E \quad \mid \quad E^\omega.$$

We define an interpretation $\inter{\ \cdot\ }$ of $\omega$-expressions into fast prostochastic words:
\begin{itemize}
	\item $\inter{a}$ is prostochastic word induced by the constant sequence of the one letter word $a$, 
	\item $\inter{E_1 \cdot E_2} = \inter{E_1} \cdot \inter{E_2}$,
	\item $\inter{E^\omega} = \inter{E}^\omega$.
\end{itemize}

The following definition of polynomial prostochastic words is in one-to-one correspondence with the definition
of polynomial sequences of finite words.

\begin{definition}{(Polynomial prostochastic word)}
The set of \textit{polynomial prostochastic words} is $\set{\inter{E} \mid E \textrm{ is an } \omega\textrm{-expression}}$.
\end{definition}

\begin{remark}
\label{rem:polynomial_prostochastic_words}
%The set of polynomial prostochastic words is included in the set of fast prostochastic words.
Why the term polynomial?

Consider an $\omega$-expression $E$, say $(a^\omega b)^\omega$, 
and the prostochastic word $\inter{(a^\omega b)^\omega}$,
which is induced by the sequence of finite words $((a^{f(n)} b)^{f(n)})\nN$.
The function $f$ grows linearly, so this sequence represents a polynomial behaviour.
Furthermore, the proofs above yield the following robustness property:
all converging sequences of finite words $((a^{g(n)} b)^{h(n)})\nN$, where $g,h : \N \to \N$
are subexponential functions, are equivalent, so they induce the same polynomial prostochastic word $\inter{(a^\omega b)^\omega}$.
We say that a function $g : \N \to \N$ is subexponential if for all constants $C > 1$ we have $\lim_n g(n) \cdot C^{-n} = 0$;
all polynomial functions are subexponential.

This justifies the terminology; we say that the polynomial prostochastic words represent all polynomial behaviours.
\end{remark}

%% file: proof.tex
In this subsection, we prove Theorem~\ref{thm:powers_stochastic_matrix}.

\begin{figure}[ht]
\begin{center}
\includegraphics[scale=.78]{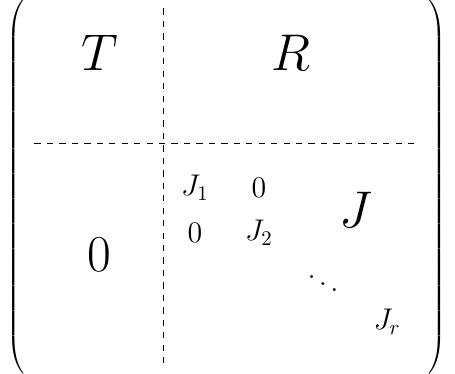}
\includegraphics[scale=.78]{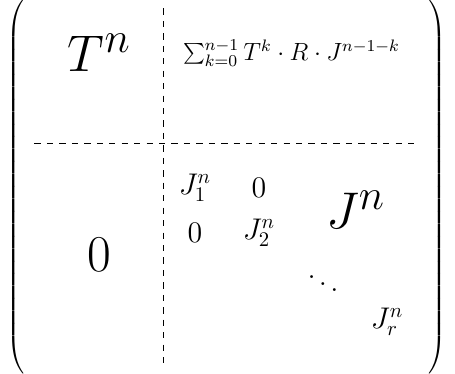}
\caption{\label{fig:matrix_decomposition} Decomposition of $P$ (on the left) and of $P^n$ (on the right).}
\end{center}
\end{figure}

Let $M \in \matS{\R}$, consider $P = M^{|Q|!}$. It is easy to see that $\boo{P}$ is idempotent.
We decompose $P$ as illustrated in Figure~\ref{fig:matrix_decomposition}, by indexing states in the following way:
\begin{itemize}
	\item first, $\boo{P}$-transient states,
	\item then, $\boo{P}$-recurrent states, grouped by recurrence class.
\end{itemize}

In this decomposition, we have the following properties:
\begin{itemize}
	\item for all $m$-transient states $s \in Q$, we have $\sum_{t\ m\textrm{-transient}} T(s,t) < 1$,
	so $\norm{T} < 1$,
	\item the matrices $J_i$ are irreducible: for all states $s,t \in Q$ corresponding to the same $J_i$,
	we have $J_i(s,t) > 0$.
\end{itemize}

The power $P^n$ of $P$ is represented in Figure~\ref{fig:matrix_decomposition}.

This decomposition allows to treat separately the three blocks:
\begin{enumerate}
	\item the block $T^n$: thanks to the observation above $\norm{T} < 1$, which combined with $\norm{T^n} \le \norm{T}^n$
	implies that $(T^n)\nN$ converges to $0$ exponentially fast,
	\item the block $\sum_{k = 0}^{n-1} T^k \cdot R \cdot J^{n-1-k}$,
	\item the block $J^n$: it is handled by Lemma~\ref{lem:power_irreducible_stochastic_matrix}.
\end{enumerate}

We first focus on item 3., and show that the sequence $(J^n)\nN$ converges exponentially fast.
Each block $J_i$ is handled separately by the following lemma.

\begin{lemma}{(Powers of an irreducible stochastic matrix)}
\label{lem:power_irreducible_stochastic_matrix}
Let $J \in \matS{\R}$ be irreducible: for all states $s,t \in Q$, we have $J(s,t) > 0$.
Then the sequence $(J^n)_{n \in \N}$ converges exponentially fast to a matrix $J^\infty$.

Furthermore, $J^\infty$ is irreducible.
\end{lemma}

This lemma is a classical result from Markov chain theory,
sometimes called the Convergence Theorem; see for instance~\cite{LevinPeresWilmer08}.

We now consider item 2., and show that the sequence $(\sum_{k = 0}^{n-1} T^k \cdot R \cdot J^{n-1-k})\nN$ converges exponentially fast.
Observe that since $\norm{T} < 1$, the matrix $I - T$ is invertible, where $I$ is the identity matrix with the same dimension as $T$.
Denote $N = (I - T)^{-1}$, it is equal to $\sum_{k \ge 0} T^k$.
Denote $J^\infty = \lim_n J^n$, which exists thanks to Lemma~\ref{lem:power_irreducible_stochastic_matrix}.

We have:
\begin{eqnarray*}
\lefteqn{\norm{\sum_{k = 0}^{n-1} T^k \cdot R \cdot J^{n-1-k} - N \cdot R \cdot J^\infty}} \\
&& =
\norm{\sum_{k = 0}^{n-1} \left[T^k \cdot R \cdot \left(J^{n-1-k} - J^\infty\right) + T^k \cdot R \cdot J^\infty\right] - N \cdot R \cdot J^\infty} \\ 
&& =
\norm{\sum_{k = 0}^{n-1} T^k \cdot R \cdot \left(J^{n-1-k} - J^\infty\right) + \left(\sum_{k = 0}^{n-1} T^k - N\right) \cdot R \cdot J^\infty} \\ 
&& \le  
\norm{\sum_{k = 0}^{n-1} T^k \cdot R \cdot \left(J^{n-1-k} - J^\infty\right)} + \norm{\left(\sum_{k = 0}^{n-1} T^k - N\right) \cdot R \cdot J^\infty}.
\end{eqnarray*}

We first consider the right summand:
\begin{eqnarray*}
\lefteqn{\norm{\left(\sum_{k = 0}^{n-1} T^k - N\right) \cdot R \cdot J^\infty}} \\
&& =  
\norm{\left(\sum_{k \ge n} T^k\right) \cdot R \cdot J^\infty} \\
&& \le
\norm{\sum_{k \ge n} T^k} \cdot \underbrace{\norm{R}}_{\le 1} \cdot \underbrace{\norm{J^\infty}}_{= 1} \\
&& =
\norm{T^n \cdot N} \\
&& \le
\norm{N} \cdot \norm{T}^n.
\end{eqnarray*}
The first equality follows from the fact that $\sum_{k = 0}^{n-1} T^k - N = \sum_{k \ge n} T^k$.
Thus, this term converges exponentially fast to $0$.

We next consider the left summand.
Thanks to Lemma~\ref{lem:power_irreducible_stochastic_matrix},
there exists a constant $C > 1$ such that for all $p \in \N$, we have $\norm{J^p - J^\infty} \le C^{-p}$.

\begin{eqnarray*}
\lefteqn{\norm{\sum_{k = 0}^{n-1} T^k \cdot R \cdot \left(J^{n-1-k} - J^\infty\right)}} \\
&& \le
\sum_{k = 0}^{n-1} \norm{T}^k \cdot \underbrace{\norm{R}}_{\le 1} \cdot \norm{J^{n-1-k} - J^\infty} \\
&& \le
\sum_{k = 0}^{n-1} \norm{T}^k \cdot \norm{J^{n-1-k} - J^\infty} \\
&& =
\sum_{k = 0}^{\lfloor n/2 \rfloor} \underbrace{\norm{T}^k}_{\le 1} \cdot \norm{J^{n-1-k} - J^\infty} + 
\sum_{k = \lfloor n/2 \rfloor + 1}^{n-1} \norm{T}^k \cdot \underbrace{\norm{J^{n-1-k} - J^\infty}}_{\le 2} \\
&& \le
\frac{C^{-(\lfloor n/2 \rfloor + 1)} - C^{-n}}{1 - C} +
2 \cdot \frac{\norm{T}^{\lfloor n/2 \rfloor + 1} - \norm{T}^n}{1 - \norm{T}} \\
&& \le
2 \cdot \left(\frac{C^{-(\lfloor n/2 \rfloor + 1)}}{1 - C} + \frac{\norm{T}^{\lfloor n/2 \rfloor + 1}}{1 - \norm{T}}\right).
\end{eqnarray*}
Thus, this term converges exponentially fast to $0$.

We proved that $(P^n)\nN$ converges exponentially fast to a matrix $M^\omega$.
We conclude the proof of Theorem~\ref{thm:powers_stochastic_matrix} by observing that:
$$\boo{M^\omega}(s,t) = 
\begin{cases}
1 & \textrm{if } \boo{P}(s,t) = 1 \textrm{ and } t \textrm{ is } \boo{P}\textrm{-recurrent,} \\
0 & \textrm{otherwise.}
\end{cases}$$

%Again, we reason independently for each block. The only non-trivial case is for the top right block,
%so we consider a $\boo{P}$-transient state $s \in Q$ and a $\boo{P}$-recurrent state $t \in Q$.

Assume first that $\boo{M^\omega}(s,t) = 1$, \textit{i.e.} $M^\omega(s,t) > 0$. 
This already implies that $t$ is $\boo{P}$-recurrent, looking at the decomposition of $P^n$.
Since $M^\omega = \lim_n P^n$, it follows that for $n$ large enough, we have $P^n(s,t) > 0$.
The matrix $\boo{P}$ is idempotent, so we have for all $n \in \N$ the equality $\boo{P^n} = \boo{P}$,
implying that $P(s,t) > 0$, \textit{i.e.} $\boo{P}(s,t) = 1$.

Conversely, assume that $\boo{P}(s,t) = 1$ and $t$ is $\boo{P}$-recurrent.
Observe that for all $n \in \N$ we have $P^{n+1}(s,t) \ge P(s,t) \cdot P^n(t,t)$.
For $n$ tending to infinity, this implies $M^\omega(s,t) \ge P(s,t) \cdot M^\omega(t,t)$.
Note that $P(s,t) > 0$, and $M^\omega(t,t) > 0$ since $t$ is $\boo{P}$-recurrent and thanks to Lemma~\ref{lem:power_irreducible_stochastic_matrix}.
It follows that $M^\omega(s,t) > 0$, \textit{i.e.} $\boo{M^\omega}(s,t) = 1$.

%% file: reformulation_characterisation.tex
For proof purposes, we give an equivalent presentation of the Markov Monoid through $\omega$-expressions.
Given a probabilistic automaton $\A$, we define an interpretation $\tr{\cdot}$ of $\omega$-expressions into Boolean matrices:
\begin{itemize}
	\item $\tr{a}$ is $\boo{\phi(a)}$, 
	\item $\tr{E_1 \cdot E_2}$ is $\tr{E_1} \cdot \tr{E_2}$,
	\item $\tr{E^\omega}$ is $\tr{E}^\sharp$, 
	only defined if $\tr{E}$ is idempotent.
\end{itemize}
Then the Markov Monoid of $\A$ is $\set{\tr{E} \mid E \textrm{ an } \omega\textrm{-expression}}$.

\bigskip
The following theorem is a reformulation of Theorem~\ref{thm:characterisation_mma},
using the prostochastic theory.
It clearly implies Theorem~\ref{thm:characterisation_mma}: indeed, a polynomial prostochastic word induces a polynomial sequence,
and vice-versa.

\begin{theorem}{(Characterisation of the Markov Monoid algorithm)}
\label{thm:characterisation_mma_prostochastic}
The Markov Monoid algorithm answers ``YES'' on input $\A$
if, and only if,
there exists a polynomial prostochastic word accepted by $\A$.
\end{theorem}

The proof relies on the notion of reification,
used in the following proposition, from which follows Theorem~\ref{thm:characterisation_mma_prostochastic}.

\begin{definition}{(Reification)}
Let $\A$ be a probabilistic automaton.

A sequence $(u_n)\nN$ of words reifies a Boolean matrix $m$ if
for all states $s,t \in Q$, the sequence $\left(\prob{\A}(s \xrightarrow{u_n} t)\right)\nN$ converges and:
$$m(s,t) = 1 \iff \lim_n \prob{\A}(s \xrightarrow{u_n} t) > 0.$$
\end{definition}

\begin{proposition}{(Characterisation of the Markov Monoid algorithm)}
\label{prop:characterisation}
For every $\omega$-expression $E$,
for every $\phi : A \to \matS{\R}$, we have 
$$\boo{\widehat{\phi}(\overline{E})} = \tr{E}.$$

Consequently, for every probabilistic automaton $\A$:
\begin{itemize}
	\item any sequence inducing the polynomial prostochastic word $\overline{E}$ reifies $\tr{E}$,
	\item the element $\tr{E}$ of the Markov Monoid is a value $1$ witness
if, and only if, the polynomial prostochastic word $\overline{E}$ is accepted by $\A$.
\end{itemize}
\end{proposition}

\begin{proof}
We prove the first part of Proposition~\ref{prop:characterisation} by induction on the $\omega$-expression $E$,
which essentially amounts to gather the results from Section~\ref{sec:prostochastic}.

\medskip
The base case of $a \in A$ is clear.

\textbf{The product case:} let $E = E_1 \cdot E_2$, and $\phi : A \to \matS{\R}$.

By definition $\overline{E} = \overline{E_1} \cdot \overline{E_2}$,
so $\widehat{\phi}(\overline{E}) = \widehat{\phi}(\overline{E_1}) \cdot \widehat{\phi}(\overline{E_2})$
because $\widehat{\phi}$ is a morphism,
and $\boo{\widehat{\phi}(\overline{E})} = \boo{\widehat{\phi}(\overline{E_1})} \cdot \boo{\widehat{\phi}(\overline{E_2})}$.
Also by definition, we have $\tr{E} = \tr{E_1} \cdot \tr{E_2}$,
so $\boo{\widehat{\phi}(\overline{E_1 \cdot E_2})} = \tr{E_1 \cdot E_2}$.
 
\textbf{The iteration case:} let $E = F^\omega$, and $\phi : A \to \matS{\R}$.

By definition, $\overline{E} = \overline{F}^\omega$,
so $\widehat{\phi}(\overline{E}) = \widehat{\phi}(\overline{F}^\omega)$,
which is equal to $\widehat{\phi}(\overline{F})^\omega$ thanks to Lemma~\ref{lem:limit_fast}.
Now, $\boo{\widehat{\phi}(\overline{F})^\omega} = \boo{\widehat{\phi}(\overline{F})}^\sharp$ thanks to Theorem~\ref{thm:powers_stochastic_matrix}.
By induction hypothesis, $\boo{\widehat{\phi}(\overline{F})} = \tr{F}$,
so $\boo{\widehat{\phi}(\overline{F}^\omega)} = \tr{F^\omega}$.

\bigskip
We prove the second part.
Consider a sequence $\mathbf{u}$ inducing the polynomial prostochastic word $\overline{E}$.
Thanks to the first item, $\boo{\widehat{\phi}(\overline{E})} = \tr{E}$,
implying that $\boo{\lim \phi(\mathbf{u})} = \tr{E}$, which means that $\mathbf{u}$ reifies $\overline{E}$.

\bigskip
We prove the third part.

Assume that $\tr{E}$ is a value $1$ witness, \textit{i.e.} for all states $t \in Q$, if $\tr{E}(q_0,t) = 1$, then $t \in F$.
So for $t \notin F$, we have $\lim \phi(\mathbf{u})(q_0,t) = 0$.
Since we have $\lim \phi(\mathbf{u})(q_0,Q) = 1$, 
it follows that $\lim \phi(\mathbf{u})(q_0,F) = 1$,
so $\widehat{\phi}(\overline{E})(q_0,F) = 1$, 
\textit{i.e.} the polynomial prostochastic word $\overline{E}$ is accepted by $\A$.

Conversely, assume that the polynomial prostochastic word $\overline{E}$ is accepted by $\A$.
Since it is induced by $\mathbf{u}$, it follows that $\lim \phi(\mathbf{u})(q_0,F) = 1$.
Consider a state $t \in Q$ such that $\tr{E}(q_0,t) = 1$.
It follows that $\boo{\lim \phi(\mathbf{u})}(q_0,t) = 1$, so $\lim \phi(\mathbf{u})(q_0,t) > 0$.
Since $\lim \phi(\mathbf{u})(q_0,F) = 1$, this implies that $t \in F$, hence $\tr{E}$ is a value $1$ witness.
\hfill\qed\end{proof}

%% file: undecidability.tex
\begin{theorem}{(Undecidability of the two-tier value $1$ problem)}
\label{thm:two_tier}
The following problem is undecidable: given a probabilistic automaton $\A$,
determine whether there exist two finite words $u,v$ such that $\lim_n \prob{\A}((u \cdot v^n)^{2^n}) = 1$.
\end{theorem}

\begin{figure}[ht]
\begin{center}
\includegraphics[scale=.76]{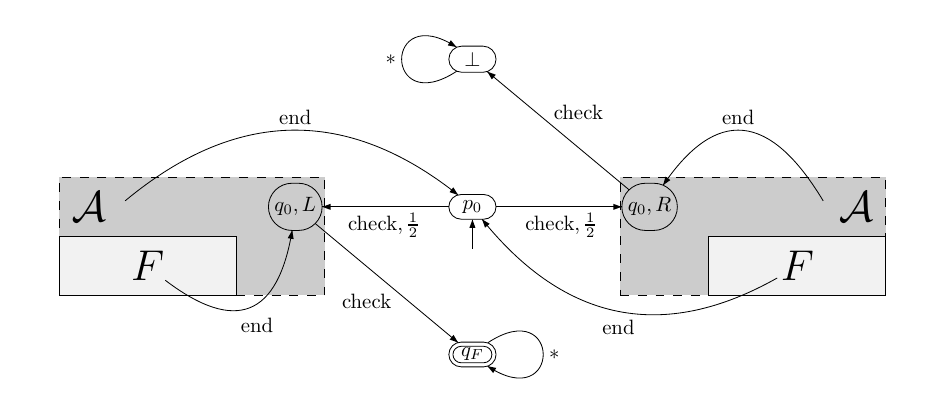}
\caption{\label{fig:reduction} Reduction.}
\end{center}
\end{figure}

The two-tier value $1$ problem seems easier than the value $1$ problem as it restricts the set of sequences of finite words to very simple sequences.
We call such sequences two-tier, because they exhibit two different behaviours: the word $v$ is repeated a linear number of times, namely $n$,
while the word $u \cdot v^n$ is repeated an exponential number of times, namely $2^n$.

The proof is obtained using the same reduction as for the undecidability of the value $1$ problem, from~\cite{GimbertOualhadj10},
with a refined analysis.

\begin{proof}
We construct a reduction from the emptiness problem for probabilistic automata to the two-tier value $1$ problem.
For technical reasons, we will assume that the probabilistic automata have transition probabilities $0,\frac{1}{2}$, or $1$.

Let $\A$ be a probabilistic automaton.
We construct a probabilistic automaton $\B$ such that the following holds:
\begin{center}
there exists a finite word $w$ such that $\prob{\A}(w) > \frac{1}{2}$ if, and only if,\\
there exist two finite words $u,v$ such that $\lim_n \prob{\B}((u \cdot v^n)^{2^n}) = 1$.
\end{center}
The emptiness problem for probabilistic automata has been shown undecidable in~\cite{Paz71}.
We refer to~\cite{GimbertOualhadj10} for a simple proof of this result.

Without loss of generality we assume that the initial state $q_0$ of $\A$ has no ingoing transitions.

The alphabet of $\B$ is $B = A \uplus \set{\chck,\ed}$,
its set of states is $Q_\B = Q \times \set{L,R}\ \uplus\ \set{p_0,\bot,q_F}$,
its transition function is $\phi'$,
the only initial state is $p_0$ and the only final state is $q_F$.
We describe $\phi'$ as a function $\phi' : Q_\B \times B \to \D(Q_\B)$:
$$\begin{cases}
%\phi'(p_0,a) & = p_0 \textrm{ for } a \in A \\
%\phi'(p_0,\ed) & = p_0 \\
\phi'(p_0,\chck) & = \frac{1}{2} \cdot (q_0,L) + \frac{1}{2} \cdot (q_0,R) \\
\phi'((q,d),a) & = (\phi(q,a),d) \textrm{ for } a \in A \textrm{ and } d \in \set{L,R}\\
\phi'((q_0,L),\chck) & = q_F \\
\phi'((q,L),\ed) & = q_0 \textrm{ if } q \in F \\
\phi'((q,L),\ed) & = p_0 \textrm{ if } q \notin F \\
\phi'((q_0,R),\chck) & = \bot \\
\phi'((q,R),\ed) & = p_0 \textrm{ if } q \in F \\
\phi'((q,R),\ed) & = q_0 \textrm{ if } q \notin F \\
\phi'(q_F,*) & = q_F
%\phi'(\bot,*) & = \bot, \\
\end{cases}$$
where as a convention, if a transition is not defined, it leads to $\bot$.

Assume that there exists a finite word $w$ such that $\prob{\A}(w) > \frac{1}{2}$,
then we claim that $\lim_n \prob{\B}((\chck \cdot (w \cdot \ed)^n)^{2^n}) = 1$.
Denote $x = \prob{\A}(w)$.

We have
\[
\prob{\A}(p_0 \xrightarrow{\chck \cdot (w \cdot \ed)^n} (q_0,L)) = \frac{1}{2} \cdot x^n,
\]
and
\[
\prob{\A}(p_0 \xrightarrow{\chck \cdot (w \cdot \ed)^n} (q_0,R)) = \frac{1}{2} \cdot (1-x)^n.
\]

We fix an integer $N$ and analyse the action of reading $(\chck \cdot (w \cdot \ed)^n)^N$: 
there are $N$ ``rounds'',
each of them corresponding to reading $\chck \cdot (w \cdot \ed)^n$ from $p_0$.
In a round, there are three outcomes: 
winning (that is, remaining in $(q_0,L)$) with probability $p_n = \frac{1}{2} \cdot x^n$,
losing (that is, remaining in $(q_0,R)$) with probability $q_n = \frac{1}{2} \cdot (1-x)^n$, 
or going to the next round (that is, reaching $p_0$) with probability $1 - (p_n + q_n)$.
If a round is won or lost, then the next $\chck$ leads to an accepting or rejecting sink; otherwise it goes on to the next round, for $N$ rounds. 
Hence:
\begin{eqnarray*}
\lefteqn{\prob{\A}((\chck \cdot (w \cdot \ed)^n)^N)} \\
&& 
= \sum_{i = 1}^{N-1} (1 - (p_n + q_n))^{i-1} \cdot p_n \\
&& 
= p_n \cdot \frac{1 - (1 - (p_n + q_n))^{N-1}}{1 - (1 - (p_n + q_n))} \\
&& 
= \frac{1}{1 + \frac{q_n}{p_n}} \cdot \left(1 - (1 - (p_n + q_n))^{N-1}\right)
\end{eqnarray*}

First, $\frac{q_n}{p_n} = (\frac{1-x}{x})^n$ converges to $0$ as $n$ goes to infinity since $x > \frac{1}{2}$.

Denote $N = f(n)$ and consider the sequence $((1 - (p_n + q_n))^{N-1})\nN$;
if $(f(n) \cdot x^n)\nN$ converges to $\infty$, then the above sequence converges to $0$.
Since $x > \frac{1}{2}$, this holds for $f(n) = 2^n$.
It follows that the acceptance probability converges to $1$.
Consequently: 
$$\lim_n \prob{\A}((\chck \cdot (w \cdot \ed)^n)^{2^n}) = 1.$$

\vskip1em
Conversely, assume that for all finite words $w$, we have $\prob{\A}(w) \le \frac{1}{2}$.
We claim that every finite word in $B^*$ is accepted by $\B$ with probability at most $\frac{1}{2}$.
First of all, using simple observations we restrict ourselves to words of the form
$$w = \chck \cdot w_1 \cdot \ed \cdot w_2 \cdot \ed \cdots\ w_n \cdot \ed \cdot w',$$
with $w_i \in A^*$ and $w' \in B^*$.
Since $\prob{\A}(w_i) \le \frac{1}{2}$ for every $i$, it follows that in $\B$, after reading the last letter $\ed$ in $w$ before $w'$,
the probability to be in $(q_0,L)$ is smaller or equal than the probability to be in $(q_0,R)$.
This implies the claim.
It follows that the value of $\B$ is not $1$, and in particular for two finite words $u,v$, we have $\lim_n \prob{\B}((u \cdot v^n)^{2^n}) < 1$.
\hfill\qed\end{proof}

Note that the proof actually shows that one can replace in the statement of Theorem~\ref{thm:two_tier} the value $2^n$ by $f(n)$ 
for any function $f$ such that $f(n) \ge 2^n$.

%% file: optimality.tex
It was shown in~\cite{FijalkowGimbertKelmendiOualhadj15} that the Markov Monoid algorithm subsumes all previous known algorithms to solve the value $1$ problem.
Indeed, it was proved that it is correct for the subclass of leaktight automata, and that the class of leaktight automata strictly
contains all subclasses for which the value $1$ problem has been shown to be decidable.

\begin{figure}[ht]
\begin{center}
\includegraphics[scale=.7]{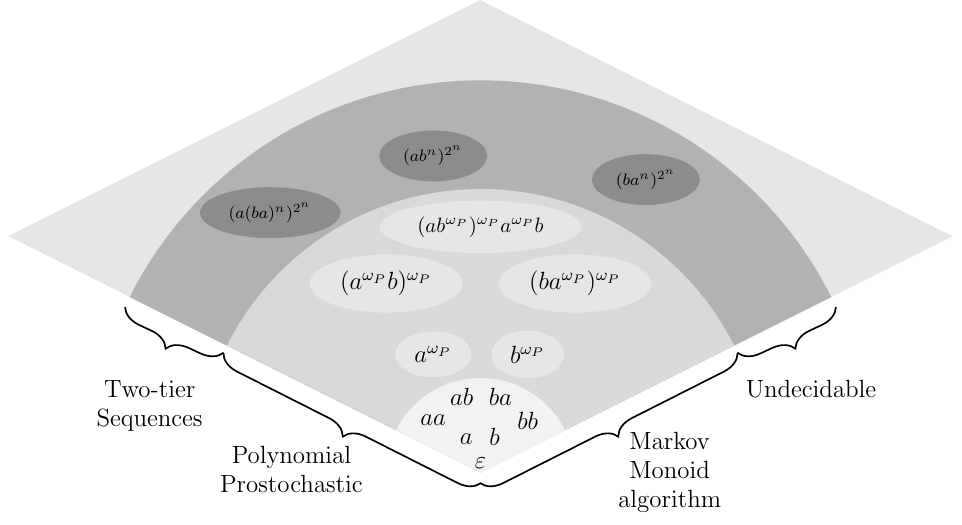}
\caption{\label{fig:free_prostochastic_monoid} Optimality of the Markov Monoid algorithm.}
\end{center}
\end{figure}

\bigskip
At this point, the Markov Monoid algorithm is the best algorithm \textit{so far}. But can we go further?
If we cannot, then what is an optimality argument? 
It consists in constructing 
a maximal subclass of probabilistic automata
for which the problem is decidable.
We can reverse the point of view, and equivalently construct 
an optimal algorithm, \textit{i.e.} an algorithm
that correctly solves a subset of the instances,
such that no algorithm correctly solves a superset of these instances.
However, it is clear that no such strong statement holds,
as one can always from any algorithm obtain a better algorithm 
by precomputing finitely many instances.

\bigskip
Since there is no strong optimality argument, we can only give a subjective argument.
We argue that the combination of our characterisation from Section~\ref{sec:mma}
and of the undecidability of the two-tier value $1$ problem from Subsection~\ref{subsec:conclusions:two-tier} 
supports the claim that the Markov Monoid algorithm is \textit{in some sense} optimal:
\begin{itemize}
	\item The characterisation says that the Markov Monoid algorithm captures exactly all polynomial behaviours.
	\item The undecidability result says that the undecidability of the value $1$ problem arises when polynomial and exponential behaviours are combined.
\end{itemize}
So, the Markov Monoid algorithm is optimal in the sense that it captures a \textit{large} set of behaviours, namely polynomial behaviours,
and that no algorithm can capture both polynomial and exponential behaviours.